\documentclass[envcountsame, envcountsect, runningheads, a4paper]{llncs}

\usepackage[utf8]{inputenc}
\usepackage{amssymb,amsmath}
\usepackage{etex}
\usepackage{wrapfig}
\usepackage{graphicx}
\usepackage[mathscr]{eucal}
\usepackage{verbatim}
\usepackage{hyperref}
\usepackage{comment}
\usepackage{paralist}
\usepackage{setspace}
\usepackage{microtype}
\usepackage{mathtools}
\usepackage{amsfonts}
\usepackage{caption}
\usepackage{color}
\usepackage{graphicx}
\usepackage{float}
\usepackage{tikz}
\usepackage{xspace} 
\usepackage[space]{cite}
\usepackage{pgf}
\usepackage[usenames,dvipsnames]{pstricks}
\usepackage{epsfig}
\usepackage{thmtools}
\usepackage{cleveref}

\renewcommand{\vec}[1]{\overrightarrow{#1}}

\def\reals{\mathbb{R}}

\def\f{{\mathbf{f}}}
\def\s{{\mathbf{s}}}
\def\c{{\mathbf{c}}}
\def\Fprop{{\mathbf{F}}}
\def\C{{\mathbf{C}}}
\def\S{{\mathbf{S}}}
\def\t{{\mathbf{t}}}
\newcommand{\resp}{{\operatorname{resp}}}

\def\set#1{{\{ #1 \}}}

\newcommand{\Plays}{\operatorname{Plays}}

\newcommand{\hist}{{\operatorname{hist}}}
\newcommand{\Post}{{\operatorname{Post}}}
\newcommand{\Act}{{\operatorname{Act}}}

\newcommand{\F}{{\mathcal{F}}}

\newcommand{\Dist}{{\operatorname{Dist}}}

\newcommand{\G}{{\mathcal{G}}}

\newcommand{\I}{{\mathcal{I}}}
\newcommand{\U}{{\mathcal{U}}}
\newcommand{\R}{{\mathcal{R}}}
\newcommand{\V}{{\mathcal{V}}}
\newcommand{\T}{{\mathcal{T}}}
\newcommand{\M}{{\mathcal{M}}}
\renewcommand{\phi}{\varphi}
\newcommand{\nature}{\mathtt{Nature}}

\usepackage{tikz}
\usetikzlibrary{positioning,arrows,automata,calc,shapes}

\definecolor{darkgreen}{rgb}{0.0, 0.5, 0.0}
\definecolor{darkred}{rgb}{0.9, 0.0, 0.0}
\definecolor{darkblue}{rgb}{0.0, 0.0, 0.9}

\tikzset{
	state/.style={rounded rectangle,draw=black,inner sep=1.5mm,minimum width=9mm,minimum height=5mm},
	rstate/.style={rectangle,draw=black,inner sep=1ex,minimum width=9mm,minimum height=5mm},
	istate/.style={minimum width=5mm, minimum height=6mm},
	bullet/.style={circle,draw=black,fill=black,inner sep=0cm, minimum size=0.7mm},
	initial text={},
	every initial by arrow/.style={->,>=stealth'},
	ptran/.style={rounded corners, ->,>=stealth',auto},
	ntran/.style={rounded corners, -,auto}
}

\title{A Game-Theoretic Account of \\ Responsibility Allocation\thanks{
		This work was funded by DFG grant 389792660 as part of TRR~248 -- CPEC (see https://perspicuous-computing.science), the Cluster of Excellence EXC 2050/1 (CeTI, project ID 390696704, as part of Germany's Excellence Strategy), DFG-projects BA-1679/11-1 and BA-1679/12-1, and the European Research Council under the
		Grant Agreement 610150 (http://www.impact-erc.eu/) (ERC Synergy Grant ImPACT).}}
\author{Christel Baier\inst{1} \and 
	Florian Funke\inst{1}\and
	Rupak Majumdar\inst{2}}
\authorrunning{C. Baier, F. Funke, and R. Majumdar}
\institute{
	Technische Universität Dresden, Dresden, Germany \and
	MPI-SWS, Kaiserslautern, Germany\\
	\email{\{christel.baier, florian.funke\}@tu-dresden.de, rupak@mpi-sws.org}
}

\begin{document}
\maketitle

\begin{abstract} 
When designing or analyzing multi-agent systems, a fundamental problem is \emph{responsibility ascription}:
to specify which agents are responsible
for the joint outcome of their behaviors and to which extent. 
We model strategic multi-agent interaction as an extensive form game of imperfect information and define notions of
forward (prospective) and backward (retrospective) responsibility.
Forward responsibility identifies the responsibility of a group of agents for an outcome along all possible plays, whereas backward responsibility identifies the responsibility along a given play.
We further distinguish between strategic and causal backward responsibility, where the former captures the epistemic knowledge of players along a play, while the latter formalizes which players -- possibly unknowingly -- caused the outcome.
A formal connection between forward and backward notions is established in the case of perfect recall.
We further ascribe \emph{quantitative} responsibility through cooperative game theory.
We show through a number of examples that our approach encompasses several prior formal accounts of responsibility attribution.

\end{abstract}

\section{Introduction}
\label{sec:intro}

The notion of \emph{responsibility} is fundamental in the study of multi-agent interaction.
Allocation of responsibility is a means by which we regulate interactions in society,
by declaring whether an action by a person in an interactive setting should be praised or blamed.
In multi-agent interactions, ascertaining who is to be held responsible and by which degree can be difficult;
thus, there is a need for formal frameworks for responsibility allocation.

We work in the framework of ``folk ethics'' conception of moral responsibility \cite{BrahamVanHees2012,BrahamVanHees2018}.
In this setting, the locus of responsibility resides with the individual rather than a collective and purely on her actions
and their consequences rather than identities, attitudes, norms, or values.
A person is ascribed responsibility for a given outcome if three conditions are met.
First, the person has \emph{agency}: they are able to plan and act intentionally and can distinguish
the setting of the interaction and outcomes.
The second is \emph{causal relevance}: there is a causal link between the actions and the outcome---to be elaborated below.
The third is the \emph{possibility to act otherwise}: 
We use our evaluation of actual and potential actions and its consequences on the outcome. 
However, our approach is descriptive in the sense that we do not analyze normative aspects such as 
values, virtues, intent, or morality underlying the actions of the agents \cite{Scanlon98}, which is also why we avoid the commonly used term \emph{blameworthiness}.

In this paper, we provide a game-theoretic account of responsibility allocation in a multi-agent interaction setting.
We model multi-agent interaction as 
a game of imperfect information in extensive form \cite{Kuhn,Owen} between $n$ individually rational players.
%
Players are assumed to be rational in a weak sense: they are aware of the game, the other players, their own actions, and the outcome.
In particular, we shall assume agency. 
We study both \emph{forward} and \emph{backward} notions of responsibility \cite{vandePoel11}.
A forward notion looks at the game as a whole and ascribes responsibilities to players based on all potential plays.
A backward notion looks at a play and ascribes responsibility to each player for that play.
We approach responsibility allocation to individual agents in two steps.
First, we look at coalitions of players and define notions of forward and backward responsibilities for a coalition. 
Second, we define a value function from coalitions to their responsibilities, and define the individual allocation as a power index of this value function \cite{Owen}.

For the first step, we distinguish between \emph{causal} and \emph{strategic} backward responsibility, which correspond to the \emph{responsiblity-as-cause} and \emph{responsibility-as-capacity} notions of \cite{vandePoel11}. 
Intuitively, a coalition is causally backward responsible for an outcome along a play if there is a different
strategy the coalition could have adopted that would have avoided the outcome, against the same strategy of the other players.
Strategic backward responsibility strengthens the requirement: the coalition should be aware, given their epistemic state,
of this ability to affect the outcome.
Our key technical result is a relationship between forward and strategic backward responsibility: in a game of perfect recall,
a coalition is forward responsible for an outcome \emph{if{}f} it contains a strategically backward responsible coalition for every play with the outcome, and is minimal with respect to this property (\Cref{thm: forward - backward}). Moreover, we show that all forms of responsibility of a coalition can be checked in polynomial time (\Cref{thm:complexity}).

For the second step, our approach follows \cite{vonNeumann1928}, and relies on a transition from non-cooperative
games to cooperative values.
While there are different measures of value in a cooperative game, we pick the Shapley value \cite{Shapley1953} for its familiarity and its canonicity.
Thus, we ascribe \emph{quantitative} measures of responsibility in order to compare the relative responsibility of agents for an outcome.
As a motivation we show that our modeling choices allow us to precisely talk about various aspects of many well-known scenarios from the 
moral philosophy and causality literature.

The condition for causal dependence is often tested in a framework of \emph{actual causality} \cite{HalpernP04,Halpern15}  (henceforth HP, after its proponents) or 
the NESS test 
 \cite{HartHonore1959}. HP provides a formal model for causal contribution to an outcome and models institutional rules as \emph{structural equations}. 
 We believe that structural equations are too weak to naturally model several situations of interest, such as agency, knowledge, or temporal
sequentiality.
Consider the prototypical example in which Suzy and Billy both throw stones at a bottle. Suzy's stone hits first and the bottle breaks \cite{HalpernBook}.
Who should be responsible?
The description of the problem or the structural equations do not clarify if either Suzy or Billy knew if the other threw a rock,
or if Billy's strategy to throw was conditional, knowing that Suzy had indeed thrown already. 
These nuances are easily modeled in our setting due to the additional modeling power of extensive form games.

In a technical sense, we can embed structural equations into our framework, and our causal backward responsibility is exactly the \emph{but-for} condition in causal reasoning (\Cref{thm:HP}). 
HP's actual causality goes beyond but-for-causes: in the above example, HP considers Suzy to be responsible -- but not Billy -- by interposing a counterfactual world in which Suzy does \emph{not} throw, Billy does throw, but still Billy's stone does \emph{not} hit the bottle (where we use Halpern's modified version \cite{Halpern15}). 
We find this problematic: the intervention (Billy's stone not hitting the bottle) is not an agent to whom we can ascribe agency. 
Because of the symmetry of the players, and our insistence on comparing strategies with other agent strategies, we hold both equally responsible. 
We believe HP's allocation of backward responsibility solely to Suzy is not uncontroversial: replace Suzy and Billy with two assassins who simultaneously
(and without knowledge of the other) shoot a person.
Even if the laws of physics decide which bullet reaches first, in moral or legal considerations, we would hold both assassins responsible.
A second advantage of games over structural equations is that responsibility of a coalition can be computed in polynomial time; thus, 
checking if an individual is responsible is in NP, as opposed to the harder class $D^P$ for HP.

\medskip
\noindent\textbf{Other related work.}
Close to our work, \cite{BrahamVanHees2012} give an account of moral responsibility as \emph{normal form} games with pure strategies (the model is one-shot, perfect-information).
Models of (qualitative and quantitative) responsibility have been studied for non-probabilistic Kripke structures with games of possibly
infinite duration (see, e.g., \cite{ChocklerHK2008,BeerBCOT09,BullingD13,YazdanpanahD16}). 
Quantitative measures of \emph{influence} in causal models \cite{ChocklerH04} have generated a fruitful strand in the causality literature 
\cite{AleksandrowiczCHI14,DBLP:journals/corr/Chockler16,DBLP:conf/aaai/FriedenbergH19}. A detailed comparison with other accounts of responsibility is the content of \Cref{sec: other work}.

Shapley-like values have been used to allocate responsibility \cite{DBLP:conf/aaai/FriedenbergH19,YazdanpanahDJAL19},
and recently been rediscovered for the explanation of machine learning models \cite{StrumbeljK14,LL17}.

\section{Preliminaries}

In this section we recall some basic definitions from non-cooperative game theory (see, e.g., \cite{Owen}). We rely on von Neumann's framework of extensive form games \cite{vonNeumann1928} with the only exception that leaves are not labeled with payoff vectors, but a binary variable $E$ or $\neg E$ indicating whether a certain event has occured or not.

\begin{definition}[Extensive form game]
	A finite \emph{$n$-player game} $\G$ in \emph{extensive form} consists of the following data:
	\begin{enumerate}
		\item a finite directed tree $\T$, called the \emph{game tree}, whose vertices are called the \emph{states} of $\G$;
		\item a partition of the non-terminal states of $\T$ into sets $P_0, P_1, ..., P_n$;
		\item for each $s\in P_i$ a finite set of actions $\Act(s)$ that is in bijection with the successors of $s$;
		\item for each $s\in P_0$ a probability distribution $p(s)$ over $\Act(s)$;
		\item for each $P_i$ a partition $P_i = \dot\bigcup_{I\in\I_i} I$ into \emph{information sets} such that $\Act(s) = \Act(s')$ for all $s, s'$ belonging to the same information set $I$, and $p(s) = p(s')$ if additionally $s,s'\in P_0$;
		\item for each leaf $s$ of $\T$ a \emph{labeling} $l(s)\in \{E, \neg E\}$.
	\end{enumerate}
\end{definition}

Intuitively, the game is played as follows. 
One starts at the root $s_0$ of $\T$. 
If the current state $s$ belongs to $P_0$, then the next state is chosen by a random player called $\nature$ according to $p(s)$.
If the current state $s$ belongs to $P_i$, then player $i$ chooses a successor state by choosing an action from $\Act(s)$. 
The intended meaning of the information set $I$ is that player $i$ cannot distinguish the states in $I$ and must choose an action independent of the specific state in $I$. 
A \emph{play} in $\G$ is a path from the root to a terminal state, and the set of plays is denoted by $\Plays(\G)$. 
The plays that run through state $s$ or information set $I$ are denoted by $\Plays_s(\G)$ and $\Plays_I(\G)$. 
A play ending in a state with label $E$ is called an \emph{$E$-play}.

Throughout, we will restrict our attention to games of \emph{perfect recall}, in which, intuitively, no player forgets his own history of actions. 
This is a reasonable assumption in many real-world scenarios, especially those of short duration. It also captures the intuition in a setting of responsibility that an agent should have known information available. For a state $s$ let $\hist_\G(s) = I_0a_0I_1a_1...a_{k-1}I_k$ be the sequence of information sets visited and actions taken on the path from the root of $\T$ to $s$. Given $C\subseteq\{1,...,n\}$ let $\hist_\G(s, C)$ be the subsequence obtained from $\hist_\G(s)$ by removing each $I_j$ and action $a_j$ that is not under control of a player in $C$.

\begin{definition}[Perfect recall]\label{def:perfect recall}
A game $\G$ has \emph{perfect recall} if for each player $i$ and any two states $s, s'$ in the same information set of player $i$ we have $\hist_\G(s,\{i\}) = \hist_\G(s',\{i\})$. 
\end{definition}

In the presence of uncertainty that comes with non-singleton information sets, players may prefer to act \emph{randomly} instead of \emph{deterministically}. 
Since we are not targeting equilibria, this difference is mild for the theory developed in this paper. 
Still, Kuhn's classical theorem on outcome equivalence in games of perfect recall \cite{Kuhn} is our motivation for allowing behavioral strategies throughout.

\begin{definition}[Strategies]
	A \emph{(behavioral) strategy} for player $i$ is an element $\sigma_i = \{\sigma_I\}_{I\in \I_i} \in \prod_{I\in \I_i}\Dist(\Act(I))$. It is \emph{pure} if each $\sigma_I$ is a Dirac distribution. A \emph{strategy profile} is a set $\{\sigma_i\}_{i=1}^n$, where each $\sigma_i$ is a strategy for player $i$.
\end{definition}

A play $\rho$ is \emph{consistent} with a strategy profile $\sigma = \{\sigma_i\}$ if for every information set $I$ on $\rho$ the chosen action has positive probability in $\sigma_I$. Taking only consistent plays induces subsets $\Plays^{\sigma}_s(\G)\subseteq \Plays_s(\G)$ and $\Plays^{\sigma}_I(\G)\subseteq \Plays_I(\G)$.

\section{Responsibility in non-cooperative games}
\label{sec:qualitative}

We now identify three qualitative notions of responsibility which we present in decreasing order according to their logical strength.

\subsection{Forward responsibility}

The individual responsibility of player $i$ will be an average of the marginal contribution of $i$ to \emph{coalitional responsibility}, i.e., the responsibility of a group of players $C$. We first formalize that the players in $C$ act \emph{collaboratively}.

\begin{definition}[Game induced by a coalition]\label{def:induced game on coalition}
	Let $\G$ be an $n$-player game and $C\subseteq \{1,...,n\}$. The $2$-player game $\G_C$ is obtained from $\G$ as follows: The two players $C$ and $\overline{C}$ are in control of the states in $P_C = \bigcup_{i\in C} P_i$ and $P_{\overline{C}}=\bigcup_{i\in \{1,...,n\}\setminus C} P_i$. Two states $s, s'\in P_C$ belong to the same information set in $\G_C$ if and only if they belong to the same information set in $\G$ and $\hist_\G(s, C) = \hist_\G(s',C)$. Similarly for $\overline{C}$. The labeling of terminal states remains unchanged, as do the states and distributions in control of player $\nature$.
\end{definition}

The information sets are the coarsest refinement of the existing information sets such that $\G_C$ is a game of perfect recall. The rationale for this is that the coalition $C$ will share knowledge among its members; states that are indistinguishable by one player in $C$ can become distinguishable by $C$ because another player in the coalition can tell them apart. In order to illustrate this point, consider the following example.

\begin{example}[Matching pennies] \label{ex:epistemic state}
	Two players independently choose heads or tails, and $E$ is the event that they made opposite choices. This game is depicted in \Cref{fig:pennies}, where the notation $s:i$ expresses that state $s$ is in control of player $i$, dashed lines connect states in the same information set, and the red leaves indicate $E$-plays. If we did not refine the information sets in $\G_C$, 
the coalition $C = \{1,2\}$ would not distinguish $s_1$ and $s_2$ even though a coalition member chose the action that produces the branches to these states. Thus enforcing perfect recall in $\G_C$ is necessary to model epistemic knowledge in the coalitional setting.

\begin{figure}[t]
	\centering
	\captionsetup{labelfont={bf,up}}
	
	\begin{tikzpicture}[->,>=stealth',shorten >=1pt,auto,node distance=0.5cm, semithick]
	\input{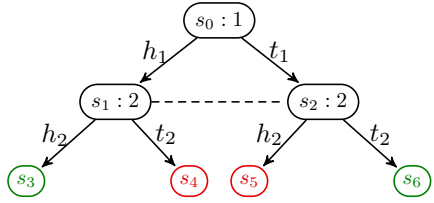}	
	\end{tikzpicture}
	\caption{Matching pennies}\label{fig:pennies}
\end{figure}
\end{example}

\begin{definition}[Forward responsibility]\label{def: forward resp}
	Let $\G$ be an $n$-player game. We say that $C\subseteq \{1,...,n\}$ is \emph{forward responsible} (henceforth $\f$-responsible) for $E$ if  {\upshape $\,$($\Fprop$)$\,$} there is a strategy $\sigma$ of the player $C$ in $\G_C$ such that all plays in $\Plays^\sigma(\G_C)$ have label $\neg E$, and $C$ is minimal with respect to property {\upshape $\,$($\Fprop$)}. 
\end{definition}

Being $\f$-responsible means that $C$ wins the game $\G_C$ under the reachability objective $\neg E$, and this condition is not satisfied for any proper subset of $C$. In order to illustrate \Cref{def: forward resp} and the subsequent notions of responsibility, we use the following toy scenario as running example. More interesting cases are postponed to the collection gathered in \Cref{sec:examples}.  

\begin{example}\label{ex:running example}
	Player 1 first decides between mode $A$ and mode $B$. In mode $A$, players 2 and 3 independently choose a side of a coin as in \Cref{fig:pennies}. In mode $B$, $\nature$ tosses a coin, the result is revealed, and then player 3 chooses a side. The variable $E$ denotes the event where the two sides of coins are not the same. See \Cref{fig:running example} for a visualization of this game. 
	
	No single player is $\f$-responsible. This is obvious for players 1 and 2, and for player 3 this is due to the non-trivial information set $\{s_3, s_4\}$: the strategy that enforces $\neg E$ in one of the states prevents it in the other. The coalition $\{1, 3\}$ is $\f$-responsible since player 1 can choose to move into mode $B$, and player $3$ responds according to what is revealed to him by $\nature$. Likewise $\{2,3\}$ is $\f$-responsible.
	\begin{figure}[h]
		\centering
		\captionsetup{labelfont={bf,up}}
		
		\begin{tikzpicture}[->,>=stealth',shorten >=1pt,auto,node distance=0.5cm, semithick]
		\input{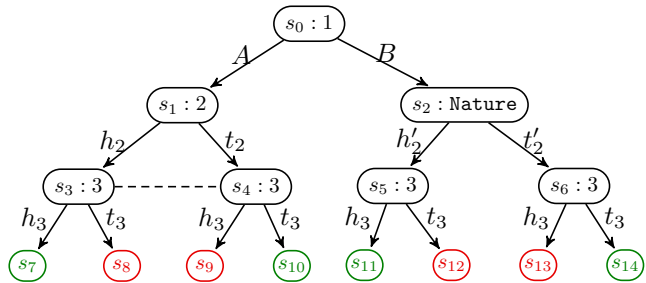}	
		\end{tikzpicture}
		\caption{$3$-player matching pennies}\label{fig:running example}
	\end{figure}
\end{example}

\subsection{Strategic backward responsibility}

In contrast to the preceeding section, backward responsibility is defined  relative to a given play $\rho$ ending in $E$. We distinguish between \emph{strategic} backward responsibility and \emph{causal} backward responsibility. Informally speaking, strategic backward responsibility means that $C$ had the power to prevent $E$ as the play $\rho$ evolved \emph{given its epistemic knowledge}: 

\begin{definition}[Strategic backward responsibility]\label{def: strategic resp}
	The set $C$ is \emph{strategically backward responsible} (henceforth $\s$-responsible) for $E$ based on an $E$-play $\rho$ if {\upshape $\,$($\S$)$\,$} there exist a state $s$ on $\rho$ and a strategy $\sigma$ of the player $C$ in $\G_C$ such that:
	\begin{enumerate}
		\item $\rho$ is consistent with $\sigma$ until $s$ is reached, and
		\item all plays in $\Plays_I^{\sigma}(\G_C)$ have label $\neg E$, where $I$ is the information set of $s$,
	\end{enumerate} 
	and $C$ is minimal with respect to property {\upshape$\,$($\S$)}.
\end{definition}

Clearly, each $\f$-responsible coalition, or even any coalition $C$ satisfying property ($\Fprop$) contains an $\s$-responsible coalition based on any $E$-play. In order to show property ($\S$) for $C$ one can simply take $s$ to be the initial state and take the strategy for $C$ that globally enforces $\neg E$.

\begin{example}\label{ex:running strategic}
	Consider again our running example depicted in \Cref{fig:running example}. First let $\rho$ be one of the two $E$-plays ending in $s_8$ and $s_{9}$. No single player is $\s$-responsible for $E$ based on one of these plays. However, the coalitions $\{1,3\}$ and $\{2,3\}$ are $\s$-responsible for $E$ based on $\rho$ since they are $\f$-responsible. Based on the two $E$-plays ending in $s_{12}$ and $s_{13}$, the single player 3 is strategically backwards responsible for $E$ since he could have chosen differently in $s_5$, respectively, $s_6$.
\end{example}

\begin{restatable}[Relating forward and backward responsibility]{theorem}{relatingFS}\label{thm: forward - backward}
	In a game of perfect recall, a coalition $C\subseteq \{1,...,n\}$ is $\f$-responsible for $E$ if and only if it contains an $\s$-responsible coalition for $E$ based on all $E$-plays, and is minimal with respect to this property.
\end{restatable}

\begin{proof}
	We show that condition ($\Fprop$) of \Cref{def: forward resp} is equivalent to the fact that $C$ contains an $\s$-responsible coalition for $E$ based on all $E$-plays. A straightforward ping-pong argument then implies that $\f$-responsibility is equivalent to this statement \emph{and} the additional minimality condition in the statement of \Cref{thm: forward - backward}.
	As was already discussed above, the forward direction is immediate.	
	
	Let us consider the backward direction. It is easy to check that if $C$ contains an $\s$-responsible coalition for $E$ based on every $E$-play, then $C$ satisfies property ($\S$) with respect to every $E$-play. Thus for every $E$-play $\rho$ there is an information set $I_\rho$ and a strategy $\sigma_\rho$ for $C$ which enforces reaching $\neg E$ from every state in $ I_\rho$. We may assume that $I_\rho$ is the first information set on $\rho$ with this property.
	Since $\G_C$ is a game of perfect recall (see \Cref{def:perfect recall} and \Cref{ex:epistemic state}), we have for two $E$-paths $\rho, \rho'$ that either $I_\rho = I_{\rho'}$ or that no path in $\G_C$ visits both $I_\rho$ and $I_{\rho'}$. For if there was such a path $\rho''$ visiting $I_\rho$ and $I_{\rho'}$ in this order (the other case is symmetric), then due to perfect recall $\rho'$ would also need to visit $I_\rho$ before $I_{\rho'}$. This is a contradiction to the assumption that $I_{\rho'}$ is the first information set on $\rho'$ from which $C$ can avoid $E$.
	
	For any information set $I$ let $\Post(I)$ denote the set of states that lie in a subtree rooted at some state in $I$. Note that $\sigma_\rho$ can be replaced by any other strategy for $C$ whose behavior on the information sets present in $\Post(I_\rho)$ is the same as $\sigma_\rho$. By the previous paragraph, no other set $I_{\rho'}$ can be present in this union. This means that the strategies $\sigma_\rho$ can be assembled into a strategy $\sigma$ such that $\sigma$ coincides with $\sigma_\rho$ on $\Post(I_\rho)$. Since every $E$-path in $\G_C$ has to pass through one of the $I_\rho$ and $\sigma$ avoids $E$ from all states contained in $I_\rho$, we conclude that $\sigma$ is a strategy for $C$ that \emph{globally} enforces reaching $\neg E$, as desired.
	\qed
\end{proof}

\begin{remark}
	\Cref{thm: forward - backward} would not hold without the refinement of information sets in $\G_C$ illustrated in \Cref{ex:epistemic state}. Consider the game of perfect recall depicted in \Cref{fig:counterexample} with $E$-plays indicated by red terminal states. If the information sets in $\G_C$ had not been refined, then $C = \{2,3\}$ would still be strategically backward responsible for $E$ based on all $E$-plays: $C$ can enforce $\neg E$ from $s_1$ and $s_2$ independently. However, in this case $C$ would \emph{not} be forward responsible since player 3 would not know which action to choose in the information set $\{s_3, s_4\}$. In the refined version $\G_C$ player $C$ \emph{can} distinguish $s_3$ and $s_4$ because player 2 can.
		
		\begin{figure}[t]
			\centering
			\captionsetup{labelfont={bf,up}}
			
			\begin{tikzpicture}[->,>=stealth',shorten >=1pt,auto,node distance=0.5cm, semithick]
			\input{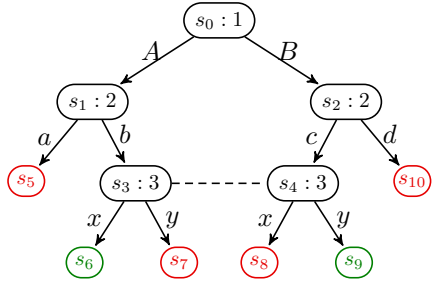}	
			\end{tikzpicture}
			\caption{Counterexample to \Cref{thm: forward - backward} without refining information sets in $\G_C$}\label{fig:counterexample}
		\end{figure}
\end{remark}

\subsection{Causal backward responsibility}

A set of players might have caused $E$ unknowingly and even inadvertently. For example, in the simple 2-player matching pennies scenario depicted in \Cref{fig:pennies}, imagine that player 1 picks heads and player 2 picks tails. Both players may be held responsible for the result since changing their individual actions would have prevented it. However, \emph{they were simply not aware of this fact}. The lack of knowledge that renders a strategic sense of responsibility implausible in this scenario refers to the uncertainty of the opponent's strategy (which also entails the uncertainty that comes with non-singleton information sets). Causal backward responsibility leverages this lack of knowledge by fixing a strategy profile for the opposing coalition. Intuitively, the coalition $C$ \emph{guesses} the strategy of $\overline{C}$ and thus attains the \emph{hypothetical} knowledge to anticipate the actions of $\overline{C}$.

\begin{definition}[Causal backward responsibility]\label{def:causal backward}
	Let $\sigma$ be a strategy profile in $\G$ such that $\Plays^{\sigma}(\G)$ contains an $E$-path $\rho$. Let $(\sigma_C, \sigma_{\overline{C}})$ be the strategy profile induced by $\sigma$ in $\G_C$. Then $C\subseteq \{1,...,n\}$ is \emph{causally backward responsible} (henceforth $\c$-responsible) for $E$ based on $\rho$ and $\sigma$ if  {\upshape $\,$($\C$)$\,$} there exists a strategy $\sigma_C'$ for $C$ in $\G_C$ such that all plays in $\Plays^{\sigma_C', \sigma_{\overline{C}}}(\G_C)$ that are also consistent with $\nature$'s random choices on $\rho$ have label $\neg E$, and $C$ is minimal with respect to property {\upshape $\,$($\C$)}. 
\end{definition}

Intuitively, $C$ satisfies condition {\upshape $\,$($\C$)$\,$} if $C$'s actions made a difference to the outcome when everything else is held fixed: The $E$-path $\rho$ consistent with the strategy profile contains a state from which the coalition $C$ \emph{could} have employed a different strategy that enforces reaching $\neg E$ \emph{provided that} $C$'s guess on the opponent's strategy was right. If, as the game evolved, $\overline{C}$ turns out to follow a different strategy, then no responsibility based on this strategy profile is ascribed -- $C$ simply made a bad guess. Fixing strategies might seem like a far-fetched setup. We illustrate in \Cref{sec:examples} that many scenarios naturally determine a strategy profile, and we show in \Cref{sec: other work} that this canonically corresponds to fixing a \emph{context} for a causal models.

\begin{example}\label{ex:example causal}
	Consider once more the scenario depicted in \Cref{fig:running example}. In the following denote the Dirac distribution over $\Act(s)$ concentrated on $a$ by $s\mapsto a$. Based on the strategy profile $\sigma^1$ given by $s_0\mapsto B, s_1\mapsto h_2, \{s_3, s_4\}\mapsto h_3, s_5\mapsto t_3, s_6\mapsto t_3$ and the play ending in $s_{12}$, players 1 and 3 are $\c$-responsible for $E$; player 1 can prevent $E$ by switching to $A$ since players 2 and 3 make identical choices in mode $A$, and player 3 can prevent $E$ by choosing heads in $s_5$. The reader is invited to check that every single player is $\c$-responsible for $E$ based on the strategy profile $\sigma^2$ given by $s_0 \mapsto A, s_1\mapsto h_2, \{s_3, s_4\}\mapsto t_3, s_5\mapsto h_3, s_6\mapsto t_3$.
\end{example}

The example illustrates that $\c$-responsibility is an approach to identify `causes' for $E$ irrespective of epistemic information available to the players. A comparison to Halpern and Pearl's actual causes is given in \Cref{sec: other work}.

 \begin{restatable}[From strategic responsibility to causal responsibility]{proposition}{strategicToCausal}\label{prop: strategic vs. causal}
 	Let $\sigma$ be a strategy profile and let $\rho\in \Plays^{\sigma}(\G)$ be an $E$-play. If $C$ is $\s$-responsible for $E$ based on $\rho$, then $C$ contains a $\c$-responsible coalition for $E$ based on $\rho$ and $\sigma$.
 \end{restatable}

 \begin{proof}
	The proposition follows from the fact that property ($\S$) of \Cref{def: strategic resp} implies property ($\C$) of \Cref{def:causal backward}, as we now show.
	By property ($\S$), there exists a state $s$ on $\rho$ together with a strategy $\sigma'$ for $C$ such that all plays in $\Plays_I^{\sigma'}(\G_C)$ have label $\neg E$, where $I$ is the information set of $s$. As $\sigma'$ is consistent with $\rho$ until $s$ is reached, we notice that $\Plays^{\sigma', \sigma_{\overline{C}}}(\G_C)\subseteq \Plays_I^{\sigma'}(\G_C)$. This shows that $\Plays^{\sigma', \sigma_{\overline{C}}}(\G_C)$ contains only paths ending in $\neg E$. Of course, this property still holds if one further restricts $\Plays^{\sigma', \sigma_{\overline{C}}}(\G_C)$ to plays on which $\nature$ only chooses the same actions as those on $\rho$, as required by property ($\C$).
	\qed
\end{proof}

\subsection{Quantifying responsibility}\label{sec:quantitative}

All responsibility notions considered so far are qualitative -- either a coalition is responsible or not. In this section we take the analysis one step further and \emph{quantify} how much responsibility each \emph{individual} player has. For this we employ the \emph{Shapley value} for cooperative games \cite{Shapley1953}, i.e., games defined by a function $g\colon 2^n\to\reals$, where $g(C)$ represents the common gain (or cost, depending on the situation) which coalition $C$ can achieve collaboratively.

The responsibility notions $\{\f, \s, \c\}$ (i.e., $\f$orward, $\s$trategically backward, $\c$ausally backward) in an extensive form game naturally induce cooperative games. Let $\t \in\{\f, \s, \c\}$; if $\t = \s$ we assume that an $E$-play $\rho$ is fixed, and if $\t = \c$, we further assume that a strategy profile $\sigma$ is fixed (which we suppress from the notation for readability). We define the \emph{induced} cooperative game $g_{\G, \t}\colon 2^n\to\reals$ by setting $g_{\G, \t}(C)$ to be $1$ if $C$ contains a $\t$-responsible coalition for $E$ (based on $\rho$, resp. $\sigma$) and to be $0$ otherwise.

\begin{definition}[Responsibility value]
	The \emph{responsibility value} $\resp_{\G, \t}(i)$ of player $i$ in an $n$-player extensive form game $\G$ is the Shapley value of $i$ in the induced cooperative game $g_{\G, \t}$, i.e., 
	\[ \resp_{\G, \t}(i) = \frac{1}{n!}\sum_{\pi \in S_n} g_{\G,\t}(\pi_{\geq i}) -  g_{\G,\t}(\pi_{\geq i}\setminus\{i\}),\]
	where $S_n$ denotes the set of permutations on $\{1,...,n\}$ and $\pi_{\geq i} = \{j\in \{1,...n\}\mid \pi(j)\geq \pi(i)\}$.
\end{definition}

We have $\resp_{\G,\t}(i) >0$ if and only if player $i$ belongs to a $\t$-responsible coalition, i.e., if $i$ leaves the coalition, then it is not $\t$-responsible anymore. Moreover, we have $\sum_i\resp_{\G,\t} (i) = \resp_{\G, \t}(\{1,...,n\})- \resp_{\G, \t}(\emptyset)$ by a simple telescope sum argument, and this value is 1 unless $C=\emptyset$ is $\t$-responsible or $\{1,\ldots, n\}$ is not $\t$-responsible. Apart from these latter exceptional cases, the responsibility value measures the relative responsibility of each player against the other players within a given game. It is not intended to be compared across different types of responsibility, across plays, or even models.

The exceptional cases above are the only instances in which our theory allows \emph{responsibility voids} \cite{BrahamVanHees2018}, i.e., $\resp_{\G,\t}(i) =0$ for every player $i$. This is not unintended:  
If $C=\emptyset$ is $\t$-responsible, then all plays must have label $\lnot E$, meaning that (as far as the event $E$ is concerned) the actions of the agents are \emph{irrelevant}. If $C=\{1,\ldots, n\}$ is not $\t$-responsible, then ensuring $\lnot E$ is \emph{impossible}, and if the game does not involve randomization, this means that no play has label $\lnot E$. We present a scenario involving responsibility voids in \Cref{ex:bogus}.

\begin{example}
	In our running example depicted in \Cref{fig:running example} we have $\resp_{\G, \f}(1) = \resp_{\G, \f}(2) = 1/6$ and $\resp_{\G, \f}(3) =2/3$. Based on a play ending in $s_8$ or $s_9$, the same responsibility value arise for $\resp_{\G,\s}(i)$. But if $\rho$ ends in $s_{12}$ or $s_{13}$, then $\resp_{\G,\s}(3) = 1$ and $\resp_{\G,\s}(i) = 0$ for $i=1,2$. Based on the strategy profile $\sigma^2$ of \Cref{ex:example causal} we have $\resp_{\G,\c}(i) = 1/3$ for every player.
\end{example}

\subsection{Computational Complexity}\label{sec:complexity}

Remarkably, the complexity of deciding $\set{\f, \s, \c}$-responsibility for a coalition is polynomial time. 
Recall that extensive form games are usually labeled with \emph{payoff vectors} that specify the outcome for each individual player.
A strategy profile induces a probability distribution over $\Plays(\G)$, and the payoff for each player is the expected payoff under this distribution. A \emph{Nash equilibrium} \cite{Nash50} is a strategy profile such that no player can improve his expected payoff solely by changing his own behavior. In a two-player zero-sum game (i.e., the entries of each payoff vector add up to $0$), the payoff in any two Nash equilibria coincide and can be found in polynomial time \cite{KollerM92,Stengel96}.

\begin{restatable}[Complexity of deciding responsibility]{theorem}{complexity}\label{thm:complexity}
	Given an $n$-player game, a coalition $C\subseteq\{1,...,n\}$, and $\t\in\{\f,\s,\c\}$. If $\t = \s$ we assume that an $E$-path $\rho$ is fixed, and if $\t = \c$, we further assume that a strategy profile $\sigma$ is fixed. It is decidable in polynomial time whether $C$ is $\t$-responsible for $E$ (based on $\rho$, resp., $\sigma$).
\end{restatable}

\begin{proof}
	We show that the properties ($\Fprop$), ($\S$), and ($\C$) are each decidable in polynomial time. The additional minimality condition necessary for $\f$-, $\s$-, and $\c$-responsibility can then be checked by deciding ($\Fprop$), ($\S$), and ($\C$), respectively, for the $|C|$ many coalitions obtained from $C$ by removing a single player. Thus we need the subsequent procedures $|C|+1\leq n+1$ many times, which is still polynomial.
	
	We begin with the case $\t=\f$. Consider $\G_C$ as a zero-sum game where the payoff of a play $\rho$ is 1 for player $C$ (and hence $-1$ for player $\overline{C}$) if and only if $l(\rho) = \neg E$. Then the payoff for player $C$ in a Nash equilibrium of $\G_C$ is equal to the probability of those plays that end in $\neg E$. Since $\G_C$ is finite, $C$ satisfies property ($\Fprop$) \emph{iff} the payoff of player $C$ in a Nash equilibrium is 1, which can be decided in polynomial time as mentioned above.
	
	Next consider the case $\t = \s$. Let $s$ be a fixed state on the given play $\rho$, and denote by $I_s$ its information set. We construct an $n+1$-player zero-sum game $\widehat{\G_C^s}$ as follows: The root is a fresh state $r$ in control of player $n+1$, and its successors are the states contained in $I_s$. From each such state, the game $\widehat{\G_C^s}$ proceeds exactly as $\G_C$, i.e., the subtree of $s'\in I_s$ in $\widehat{\G_C^s}$ is identical to the subtree of $s'$ in $\G_C$. All information sets are restricted to the union of these trees. As before the payoff of a play $\rho$ is $1$ if and only if $l(\rho) = \neg E$. Then $C$ satisfies property ($\S$) with respect $E$ and $\rho$ \emph{iff}  there exists a state $s$ on $\rho$ such that the payoff for $C$ in a Nash equilibrium of $\widehat{\G_C^s}$ is $1$. This reduces deciding ($\S$) to linearly many checks each of which takes polynomial time.
	
	Finally consider the case $\t =\c$, let $(\sigma_C, \sigma_{\overline{C}})$ be the strategy profile in $\G_C$ induced by $\sigma$, and let $\rho$ be a play consistent with it. Write $ \sigma_{\overline{C}} = \{\sigma_I\}$, where $I$ ranges over the information sets of $\overline{C}$ in $\G_C$. Let $\overline{\G_C}$ be the zero-sum game obtained from $\G_C$ by deleting an edge $s\xrightarrow{a} s'$ if and only if (1) $s\in P_{\overline{C}}$ and $ \sigma_I(a) = 0$, where $I$ is the information set of $s$, or (2) $s\in P_0$, the information set of $s$ occurs on $\rho$, but has \emph{not} taken action $a$ there.
	Assign payoff vectors in the same fashion as in the previous cases. Clearly, $\overline{\G_C}$ can be constructed in linear time, and $C$ satisfies property ($\C$) with respect to $\sigma$ and $\rho$ \emph{iff} the payoff for $C$ in a Nash equilibrium of $\overline{\G_C}$ is 1. 
	\qed
\end{proof}

A straightforward consequence of \Cref{thm:complexity} is that deciding whether $\resp_{\G, \t}(i) >0$ belongs to NP, and computing $\resp_{\G, \t}(i) $ belongs to \#P.

\section{Examples}\label{sec:examples}

We now illustrate our notion of responsibility allocation on several examples in the moral
philosophy literature (cf.\ \cite{HalpernBook}).

\begin{example}[Bystanders]\label{ex:bystanders}
	Consider a car accident with three victims in immediate need of first aid. There are four bystanders $1,2,3,4$ who arrive in this order and who can help one of the victims (we implicitly assume that if one victim is already being helped, and another bystander decides to help, then she helps an unassisted victim). Let $E$ be the event where at least one victim dies. Then a coalition is $\f$-responsible if it contains exactly three players, and we get $\resp_{\G, \f}(i) = 1/4$, which just means that they have globally the same power. Consider the play where bystanders 1 and 3 help and the others do not help. Then the only $\s$-responsible coalition is $\{4\}$. Hence $\resp_{\G,\s} (4) = 1$, which is a numerical interpretation of the fact that 4 \emph{knows} that she can (or must) make the difference for the remaining victim to survive, while player 2 does not have this certainty (one can argue that she can hope to rely on the others). 
	
	Next consider the scenario in which 1 helps, but the others do not help irrespective of any previous actions. Then $\resp_{\G,\c}(1) = 0$, $\resp_{\G,\c}(i)  = 1/3$ for $i=2,3,4$, which explains that the remaining bystanders are equally responsible for the death of the victims. Finally imagine that bystander 4 would have helped if bystander 3 had helped (the \emph{bystander effect}), but 2 and 3 never help. In this scenario we get $\resp_{\G,\c}(1) = 0, \resp_{\G,\c}(2) = \resp_{\G,\c}(4) = 1/6$, and $\resp_{\G,\c}(3)=2/3$ which puts more causal responsibility on 3's shoulders for not being an exemplar to 4.
\end{example}

\begin{example}[Democratic voting]
In a democratic, secret vote between two oppenents $A$ and $B$, the absolute majority wins. There are $n$ (for simplicity let $n$ be odd) voters who can vote $A$ or $B$. Let $E$ be the event that $A$ looses the vote. A set $C$ is $\f$-responsible if and only if $|C|= (n+1)/2$, and due to the secrecy of the vote, the same condition characterizes $\s$-responsible coalitions based on any play in which $A$ looses. Thus $\resp_{\G, \f}(i) = \resp_{\G,\s}(i) = 1/n$ for all voters. Now imagine a scenario in which $A$ wins by a $4-7$ vote (so $n=11$). Then $C$ is $\c$-responsible if and only if $|C| =2$ and the two voters voted for $B$; for if they switch to $A$, then $A$ wins. Hence $\resp_{\G,\c}(i) = 1/4$ for each voter $i$ who voted for $B$.
\end{example}

\begin{example}[Marksmen]
 	Ten marksmen form a firing squad for the execution of a prisoner. They know that exactly one of them has a live bullet in his rifle, but they do not know which one. All ten have (at least theoretically) the choice of firing or not. This is modeled as a game where $\nature$ first chooses the player $i^*$ that has the live bullet and then the marksmen concurrently shoot or not. Let $E$ be the event that the prisoner dies, and consider an $E$-play $\rho$ (i.e., $i^*$ shoots). Due the uncertainty introduced by $\nature$, only the coalition consisting of all marksmen is $\f$-responsible, respectively, $\s$-responsible for $E$ based on $\rho$, and hence $\resp_{\G, \f}(i) = \resp_{\G,\s}(i) = 1/10$. However, we always have $\resp_{\G,\c}(i^*) = 1$ independent of the strategies of the other nine marksmen -- i.e. the one with the live bullet carries the entire causal responsibility.
\end{example}

\begin{example}[Bogus prevention]\label{ex:bogus}
	A person $P$ is protected by a bodyguard $B$, who suspects a poisonous attack and puts an antidote into $P$'s coffee. Indeed, an assassin $A$ tries to poison $P$, but the poison is neutralized by the antidote. Consider the event $E$ that $P$ survives. Based on the play $\rho$ and strategies described above, no single player is $\f$- or $\s$-responsible, but $B$ is $\c$-responsible, while $A$ is not.
	
	Now consider the alternative model in which $A$, seeing $B$'s action, poisons the coffee exactly if $B$ chose to put in the antidote, which determines a strategy $\sigma_A$ (e.g., because $A$ and $B$ are old friends, and $A$ wants $B$ to be considered important by $P$). Suppose that $B$ puts in the antidote, thus determining a strategy $\sigma_B$ (inducing the same path $\rho$ as above). Based on these strategies no individual player is $\c$-responsible for $P$'s survivial! We emphasize this point since all versions of Halpern-Pearl's theory \emph{do} consider $B$ putting in the antidote an actual cause of $P$'s survival, even though the structural equations determined by this model and context already preclude $P$'s death.
\end{example}

\begin{example}[Prisoner's death]
	Consider the scenario that a prisoner dies if guard 1 loads a gun, and guard 2 shoots this gun, or guard 3 shoots an already loaded gun. Guard 2 does not know whether guard 1 loaded the gun, but guard 3 is aware of all previous decisions. There is a straightforward encoding of this in terms of a three-layered game. Let $E$ be the event that the prisoner dies. Then $\{1,3\}$ and $\{2,3\}$ are minimal coalitions that are $\f$-responsible, and hence $\resp_{\G, \f}(1) = \resp_{\G, \f}(2) = 1/6$ and $\resp_{\G, \f}(3) =2/3$.
	
	Let $\rho$ be the play, where 1 loads the gun, guard 2 does not shoot, but guard 3 shoots. Based on this play, guard 3 forms the only minimal $\s$-responsible for $E$, meaning that $\resp_{\G,\s}(3) = 1$ and formalizing that guard 3 could have prevented the prisoner's death if he had wanted to. Also, if player 3 shoots no matter what guards 1 and 2 do, then there are no other $\c$-responsible coalitions. However, if guard 3 would not shoot if guard 1 does \emph{not} load the gun and guard 2 does \emph{not} shoot (e.g., because he is of inferior rank and only tried to impress the others by 'finishing the job' which either of them inadvertently failed to accomplish), then guard 1 is also $\c$-responsible for the prisoner's death based on $\rho$ and the above strategies. Hence $\resp_{\G,\c}(1) = \resp_{\G,\c}(3) = 1/2$ which has a clear interpretation: Even though 3 shoots, he did so incited by $1$'s loading the gun.
\end{example}

\section{Comparison to other approaches}\label{sec: other work}

We now provide a detailed comparison of our notions with related approaches to responsibility allocation.

\subsection{Causal models}\label{sub:AC}

An influential formal concept of causality is Halpern and Pearl's notion of \emph{actual causes} in \emph{causal} (often also called \emph{structural}) 
models \cite{HalpernP04}. 
A causal model $\M$ is a tuple consisting of a set $\U$ of exogeneous variables, a set $\V$ of endogeneous variable, a function $\R$ mapping each variable $X\in \U\cup \V$ to a finite set $\R(X)$ of possible values, and a collection of functions $\F = (F_X)$ which associates to each $X\in \V$ a function $F_X\colon \prod_{U\in \U} \R(U) \times \prod_{Y\in \V\setminus X} \R(Y) \to \R(X)$ that explains the dependency among the variables. We make the standard assumption that the causal model is \emph{recursive} meaning that there is a total order $\prec$ on $\V$ such that $X\prec Y$ implies that $F_X$ is independent of $Y$. A \emph{context} for $\M$ is a tuple $\vec{u}$ specifying values for the variables in $\U$, which unambiguously induces values for $\V$.

A \emph{primitive event} is a formula of the form $X = x$ for some $X\in \V$ and $x\in\R(X)$, and an \emph{event} is a boolean combination of primitive events. A \emph{causal formula} is of the form $[\vec{Y} \leftarrow \vec{y}]\phi$, where $\vec{Y}\subseteq \V, \vec{y}\in \prod_{Y\in\vec{Y}} \R(Y)$, and $\phi$ is an event. The formula $[\vec{Y} \leftarrow \vec{y}](X=x)$ is true in a causal model $\M$ with context $\vec{u}$ if, roughly speaking, setting the variables $\U$ to $\vec{u}$ and the variables in $\vec{Y}$ to $\vec{y}$ makes $X$ have value $x$ in the unique solution vector of the remaining functions in $\F$. This is lifted to general causal formulas in the obvious way and written as $(\M,\vec{u})\models [\vec{Y} \leftarrow \vec{y}]\phi$. The act of forcing $\vec{Y}$ to have values $\vec{y}$ independent of the prescribed dependencies is called an \emph{intervention}.

\begin{definition}[Actual causes]\label{def:AC}
	We say that $\vec{X} =\vec{x}$ is an \emph{actual cause} of $\phi$ in $(\M,\vec{u})$ if the following three conditions hold:
	\begin{enumerate}
		\item[(AC1)] $(\M,\vec{u})\models (\vec{X}=\vec{x})$ and $(\M,\vec{u})\models \phi$;
		\item[(AC2)] There is a set $\vec{W}\subseteq\V$ and a setting $\vec{x}'$ for $\vec{X}$ such that if \linebreak $(\M,\vec{u})\models (\vec{W}=\vec{w})$, then 
		\[(\M,\vec{u})\models [\vec{X} \leftarrow \vec{x}', \vec{W}\leftarrow \vec{w}]\neg\phi\]
		\item[(AC3)] $\vec{X}$ is minimal with respect to these properties.
	\end{enumerate}
	An actual cause is a \emph{but-for cause} if the set $W$ appearing in (AC2) can be taken to be empty.
\end{definition}

In the remainder of this section we argue that causal models can naturally be encoded as extensive form games, and that our notion of $\c$-responsibility captures precisely but-for causes. 
	Let $\M = (\U,\V,\R, \F)$ be a causal model, let $X_1\prec X_2... \prec X_n$ be a total order on $\V$ as induced by $\F$. Then the game $\G(\M)$ is defined as follows: The players are precisely the variables $\V$, the states are $k$-tuples $(x_1,...,x_k) \in \prod_{i=1}^k \R(X_i)$ for $1\leq k\leq n$. There is an edge $(x_1,...,x_k)\to(x_1',...,x'_{k+1})$ if and only if $x_i=x_i'$ for all $1\leq i\leq k$. Player $X_k$ is in control of the states consisting of $(k-1)$-tuples, and $\Act((x_1,...,x_{k-1})) = \R(X_k)$. Choosing value $x\in \R(X_k)$ moves the token to the state $ (x_1,...,x_{k-1},x)$. All information sets are singeltons.
	
	A play in $\G(\M)$ corresponds to a setting $\vec{v}$ for $\V$, restrictions of which to $\vec{X}\subseteq\V$ will be denoted by $\vec{x}$. For each primitive event $X=x$ we consider the set $E(X=x)$ of terminal vertices of $\G(\M)$ that lie at the end of a play on which $X$ has chosen the action $x$. This is lifted to events $\phi$ in the obvious way, written as $E(\phi)$. 
	
	Each context $\vec{u}$ for $\M$ induces values $\vec{v}$ for $\V$, and a strategy $\sigma_{X_k}^{\vec{u}}$ for each player $X_k$ by setting $\sigma_{X_k}^{\vec{u}}((x_1,..., x_{k-1})) = F_{X_k}(\vec{u}, x_1,...,x_{k-1})$. The strategy induced from these on a coalition $C\subseteq \V$ is denoted by $\sigma_C^{\vec{u}}$. Since causal models are not subject to random behavior, there exists a unique play $\rho_{\vec{u}}$ in $\G(\M)$ that is consistent with the strategy profile $\sigma^{\vec{u}}$.

\begin{restatable}{theorem}{encodingHP}\label{thm:HP}
	For a causal model $\M$ with context $\vec{u}$, $\vec{X}\subseteq \V$ is $\c$-responsible for $E(\phi)$ based on $\rho_{\vec{u}}$ and $\sigma^{\vec{u}}$ 
	in $\G(\M)$ iff $\vec{X}=\vec{x}$ is a but-for cause for $\phi$ in $(\M, \vec{u})$.
\end{restatable}
\begin{proof}
	We show that property $(\C)$ (see \Cref{def:causal backward}) for the mentioned instance of $\c$-responsibility is equivalent to (AC1) and (AC2) with $\vec{W}=\emptyset$ in $(\M, \vec{u})$. This entails the additional minimality statement required for $\c$-responsibility by a straightforward ping-pong argument involving (AC3).
	
	We begin with the forward direction. Clearly (AC1) is satisfied by the definition of $\vec{x}$ (as part of $\vec{v}$) and since $\rho$ is an $E(\phi)$-path. Since $\vec{X}$ is $\c$-responsible for $E(\phi)$ based on $\rho_{\vec{u}}$ and $\sigma^{\vec{u}}$, there exists an alternative strategy $\sigma'$ for $C$ in $\G(\M)_C$ such that all plays in $\Plays^{\sigma',\sigma_{\overline{C}}^{\vec{u}}}(\G(\M)_C)$ have label $\neg E(\phi)$. Without loss of generality we can assume that $\sigma'$ is a pure strategy as we otherwise pick in each state $s$ some action $a$ with $\sigma'_s(a)> 0$ and define a pure strategy $\sigma''$ by $\sigma''_s(a) =1$. This process singles out a unique path $\rho'\in\Plays^{\sigma',\sigma_{\overline{C}}^{\vec{u}}}(\G(\M)_C)$ that must necessarily have label $\neg E(\phi)$. Let $\vec{v}'$ be the actions taken by the set of players $\V$ on $\rho'$, and denote its restriction to $\vec{X}$ by $\vec{x}'$. Since $\sigma_{\overline{C}}^{\vec{u}}$ is induced by the functions $F_Y$ for $Y\in\V\setminus\vec{X}$, we have $(\M, \vec{u})\models [\vec{X}\leftarrow\vec{x}'](\V = \vec{v}')$. Since $\rho'$ ends in $\neg E(\phi)$ this means that $\phi$ becomes a false statement when tested against the vector $\vec{v}'$. This implies that $(\M, \vec{u})\models [\vec{X}\leftarrow\vec{x}']\neg\phi$, which shows (AC2) with $\vec{W} = \emptyset$. 
	
	For the backward direction, assume that $\vec{X} = \vec{x}$ satisfies (AC1) and (AC2) with $\vec{W}=\emptyset$. We intend to show that $\vec{X}$ is $\c$-responsible for $E(\phi)$ in $\G(\M)$ based on $\rho_{\vec{u}}$ and $\sigma^{\vec{u}}$. Clearly, $\rho_{\vec{u}}$ is an $E(\phi)$-path since $(\M,\vec{u})\models \phi$ by (AC1). Let $\vec{x}'$ be the vector garantueed to exist by (AC2) with $\vec{W}=\emptyset$, i.e., $(\M,\vec{u})\models [\vec{X} \leftarrow \vec{x}']\neg\phi$. Consider the strategy $\sigma'$ of player $C$ in $\G(\M)_C$ that forces the variables in $\vec{X}$ to take the values $\vec{x}'$, i.e., if  $s$ is a state that in $\G(\M)$ is in control of player $X\in\vec{X}$, then $\sigma'_s(y) = 1$ if and only if $y = x'$. Then  $\Plays^{\sigma',\sigma_{\overline{C}}^{\vec{u}}}(\G(\M)_C)$ contains a single path $\rho'$ as the strategy profile is pure, and we claim that this path $\rho'$ has label $\neg E(\phi)$. Since $\sigma_{\overline{C}}^{\vec{u}}$ is induced from the structural equations, the action taken by $Y\in \V\setminus \vec{X}$ on $\rho'$ is the unique value $y$ such that $(\M, \vec{u})\models [\vec{X}\leftarrow\vec{x}'](Y=y)$. As $(\M,\vec{u})\models [\vec{X} \leftarrow \vec{x}']\neg\phi$, the path $\rho'$ has $\neg E(\phi)$, as desired. This concludes the proof that $\vec{X}$ is $\c$-responsible.
	\qed
\end{proof}

It is arguably a feature of Halpern and Pearl's framework that it can ascribe actual causality to a larger class of events than but-for causes, e.g., it can distinguish the two players in the Suzy-Billie rock-throwing example from the introduction.
It is noteworthy, though, that in this case (and many others, cf.\ \cite{Halpern15}) the structural model approach distinguishes Suzy as the 
sole actual cause based on \emph{auxiliary variables} that are \emph{not} actions of independent players, but rather intermediate consequences of these actions. 
In the rock-throwing model the witness is Billy's rock not hitting the bottle rather than Billy's throwing. 
But the fact that Billy's rock did not hit the bottle is not a decision made by Billy---it just happened on the basis of previous actions 
as dictated by the structural equations. 
Auxiliary variables do not entail \emph{agency} nor the possibility to \emph{act} otherwise (compare the introduction).  
Our intention is to locate responsibility exclusively within and against (the actions of) autonomous players and we believe our notion
is more appropriate in ascribing responsibility.

\subsection{Degree of responsibility and blame}

Chockler and Halpern assign to actual causes (in the original definition \cite{HalpernP04}) a \emph{degree of responsibility} which tries to measure how crucial the cause is to the effect \cite{ChocklerH04}. It is essentially defined as the inverse of the size of a minimal set of interventions such that swapping the cause's value flips the truth value of the effect. It is noteworthy that the degree of responsibility is designed to be comparable across models, while the responsibility value measures the responsibility of each player against the other players within a given model.

Translated to our approach, the degree of $\t$-responsibility of a player would be the inverse of the size of the smallest $\t$-responsible coalition that she belongs to. In our running example, all three players would have degree of $\f$-responsibility $1/2$, since $\{1,3\}$ and $\{2,3\}$ are $\f$-responsible. However, the fact that $3$ belongs to both is ignored. Our responsibility value, on the other hand, takes \emph{all}  $\t$-responsible coalitions into account in which the player participates. This captures the idea that belonging to many responsible coalitions makes the player less dependent on others and hence more powerful.

The blame is defined as the expected degree of responsibility subject to a given probability distribution over the set of contexts. One can envision a refined notion of responsibility value in the same spirit, where the random player $\nature$ chooses a context first, and the Shapley value is taken with respect to the responsibility in each branch.

\subsection{Concurrent game structures}

The work whose intention is closest to ours is \cite{YazdanpanahDJAL19}, where notions of forward and backward responsibility are defined in the context of concurrent epistemic game structures. However, there are several differences as outlined below.  

First, they do not consider a causal notion of backward responsibility, and their notion of backward responsibility does not take the epistemic state appropriately into account: An alternative strategy of a responsible coalition is required to bring about an alternative outcome only from a single state on the play in question, not from \emph{all} states in the same information set as outlined below. For example, in their Figure 2, the second player $E_2$ is backward responsible based on the play $q_0q_1q_4$ (he can avoid the red states from $q_1$), but not based on the play $ q_0q_2q_6$, even though he acted in \emph{exactly} the same way based on \emph{exactly} the same knowledge.

Second, their notion of backward responsibility is asymmetric: In a game with two players choosing simultaneously and independently from the same set of actions based on the same knowledge it can happen that they perform identical actions, but one is backward responsible and the other one is not. This makes sense when one examines causation (which we do in our Definition 7), but not in their strategic setting. This phenomenon already becomes apparent in the matching pennies example of \Cref{ex:epistemic state}. If the play under consideration is $\rho = s_0s_1s_4$, then the definition of \cite[Definition 3.1]{YazdanpanahDJAL19} makes player 2 backward responsible, but not player 1. In our framework, both player 1 and player 2 are $\c$-, but not $\s$-responsible.

Third, while \cite[Theorem 3.3]{YazdanpanahDJAL19} states that backward responsibility is equivalent to forward responsibility from \emph{all states} on the given play, we proved in our Theorem 1 that forward responsibility is equivalent to strategic backward responsibility on \emph{all plays}. This is an important shift of perspective that has---to the best of our knowledge---not been proved before for similar notions of responsibility. In fact, it responds to an important philosophical question---left open in prior work---that 
only \cite{NaumovT20} slightly touches with its \emph{knowledge ex post} modality: 
the relation between \emph{general} (or type-level) responsibility and \emph{specific} (or token-level) responsibility. 

Fourth, in contrast to our polynomial time result, the complexity of checking backward responsibility in the sense \cite{YazdanpanahDJAL19} of a coalition  against ATL formulas is $\Delta_2^P$-complete.

\subsection{Proof-theoretic approaches}

The work \cite{Broersen11} as well as the series of papers \cite{NaumovT20c,NaumovT20b,NaumovT20} provide proof-theoretic approaches to responsibility (there called \emph{blameworthiness}, a term we avoid as it typically involves normative features). They define modal logics with various forms of
responsibility modalities and sound and complete axiomatizations for the logical systems. 
Instead, we take an operational, model-theoretic approach.
Our intention is to formulate notions of responsibility in a game-theoretic framework that allows \emph{operational} 
modeling of various natural features (temporality, strategy, information) in an obvious way. 
In other words, we regard the simplicity of extensive form games as their major strength, which in particular renders them prone to be used in interdisciplinary research with experts from cognition and philosophy.

Moreover, the computational complexity of their formalizations is left open. Given the semantics of blameworthiness of \cite{NaumovT20}, second-order blameworthiness of \cite{NaumovT20b}, or defender's blameworthiness of \cite{NaumovT20c}, one can expect at least PSPACE-hardness (and often undecidability)
for checking the blameworthiness (of a coalition, resp., a defender) in these formalisms.
Known complexity results about (original or Chellas's) STIT formulae also point to PSPACE-hardness for the notions of \cite{Broersen11}.	
This contrasts to our \emph{polynomial time} algorithm for checking all notions of responsibility. 
The aforementioned papers do not consider an analogue of our \emph{forward responsibility} notion, and they do not provide a \emph{quantitative} version of responsibility assigned to \emph{individuals} or a way to compare responsibilities (see, e.g., our \Cref{ex:bystanders}).
This is a shortcoming in terms of the applicability of their formalisms since in both societal and legal terms we usually ascribe responsibility 
to individuals. 

\section{Conclusion}

We have argued that extensive form games provide a conceptually convenient formal framework for responsibility ascription with just the right trade-off between expressiveness and tractability. 
We have defined qualitative (coalitional) and quantitative (individual) responsibility notions.
Through a set of examples, we demonstrate that our notions capture intuitive responsibility ascription
in many subtle examples.


\bibliographystyle{splncs04}
\bibliography{references}

\end{document}